%% file: 00-main.tex
\documentclass[11pt]{article}
\usepackage[utf8]{inputenc}
\usepackage{url}
\usepackage[colorlinks]{hyperref}
\usepackage{fullpage}
\usepackage{color}
\usepackage{enumerate}
\usepackage{comment}
\usepackage{cite}
\usepackage{microtype}
 
\hypersetup{linkcolor=black,anchorcolor=black,citecolor=blue,urlcolor=black}

\usepackage{libertine}

 \usepackage{amsmath, amsthm, amsfonts, verbatim, mathabx, breqn, mathtools, graphicx, tikz, complexity}
\usepackage{subcaption}

\hypersetup{pdftitle=Competitive Online Search Trees on Trees}
\hypersetup{pdfkeywords={search trees} {dynamic optimality} {data structures}}

\newtheorem{defn}{Definition}[]
\newtheorem{obs}{Observation}[]
\newtheorem{theorem}{Theorem}[section]
\newtheorem{lemma}[theorem]{Lemma}

\newcounter{note}[section]
\newtheorem{corol}[theorem]{Corollary}

\newcommand{\vomit}[1]{}
\newcommand{\strucname}{Tango Search Tree}
\newcommand{\Steiner}{Steiner-closed}

\DeclareMathOperator{\OPT}{OPT}
\DeclareMathOperator{\ALG}{ALG}

\DeclareMathOperator{\St}{Steiner}
\DeclareMathOperator{\adj}{adj}
\DeclareMathOperator{\height}{height}
\DeclareMathOperator{\cut}{Cut}

\newcommand{\lrp}[1]{\left( #1 \right)}

\usepackage{titlesec}
\titleformat{\subsubsection}[runin]
        {\normalfont\bfseries}
        {\thesubsubsection}
        {0.5em}
        {}
        [.]

\usetikzlibrary{arrows,shapes,graphs}
\tikzstyle{vertex}=[circle,draw=black,fill=white,minimum size=10pt,inner sep=0pt]
\tikzstyle{fvertex}=[circle,draw=black,fill=blue!40,minimum size=10pt,inner sep=0pt]
\tikzstyle{svertex}=[regular polygon,regular polygon sides=4,draw=black,fill=orange!70,minimum size=10pt,inner sep=0pt]
\tikzstyle{steiner}=[diamond,draw=black,fill=red,minimum size=10pt,inner sep=0pt]
\tikzstyle{edge} = [draw,thick,-]
\tikzstyle{thick edge} = [draw=blue,ultra thick,-]
\tikzstyle{dashed edge} = [draw,thick,dashed,-]
\tikzstyle{oriented edge} = [draw,line width=1pt,->]
\tikzstyle{dashed oriented edge} = [draw,thick,dashed,->]

\title{
Competitive Online Search Trees on Trees}
\author{
Prosenjit Bose\thanks{Carleton University, Ottawa, Canada. Research supported in part by NSERC.}
\and
Jean Cardinal\thanks{Universit\'{e} libre de Bruxelles (ULB), Belgium.}
\and
John Iacono$^\dag$\thanks{Supported by the Fonds de la Recherche Scientifique-FNRS under Grant no MISU F 6001 1.}\ \ \thanks{New York University. Supported by grant NSF AitF 1533564.}
\and
Grigorios Koumoutsos$^\dag$\footnotemark[4]
\and
Stefan Langerman$^\dag$\thanks{Directeur de recherches du F.R.S-FNRS.} }

\begin{document}
\maketitle

\begin{abstract}
We consider the design of adaptive data structures for searching elements of a tree-structured space.
We use a natural generalization of the rotation-based online binary search tree model in which the underlying search space is the set of vertices of a tree. 
This model is based on a simple structure for decomposing graphs, previously known under several names including elimination trees, vertex rankings, and tubings.
The model is equivalent to the classical binary search tree model exactly when the underlying tree is a path.
We describe an online $O(\log \log n)$-competitive search tree data structure in this model, matching the
best known competitive ratio of binary search trees.
Our method is inspired by Tango trees, an online binary search tree algorithm, but critically needs several new notions including one which we call Steiner-closed search trees, which may be of independent interest. Moreover our technique is based on a novel use of two levels of decomposition, first from search space to a set of Steiner-closed trees, and secondly from these trees into paths.
\end{abstract}


\setcounter{page}{0}
\newpage

\input{01-intro.tex}

\input{02-model.tex}
\input{03-lower-bound.tex}

\input{04-upper_bound.tex}

\bibliographystyle{alpha}
\bibliography{ToT} 


\end{document}

%% file: 01-intro.tex
 \section{Introduction}

\subsection{Problem and Motivation}

Efficient retrieval of data from a structured, finite, universe is a fundamental computational question in computer science.
A classical data structure designed to search elements of a linearly ordered set is the {\em binary search tree}, in which every node contains an element $x$ of the set, and its two subtrees are binary search trees for the elements that are respectively larger or smaller than $x$.
Due to their fundamental flavor and practical importance, binary search trees are the subject of a considerable scientific literature spanning nearly eight decades of research, as well as being a basic component of software libraries.

\paragraph{Searching Vertices of a Tree.} 
We consider the problem of searching for an element that is not part of a linearly ordered set, but rather a vertex of a tree $G$.
For this purpose, we generalize binary search trees to {\em search trees on trees}. 
In such a search tree, the root node contains a vertex $x$ of $G$, and its subtrees are search trees on each of the connected components of the forest $G\setminus x$ (see Figure~\ref{fig:ex_search_tree} for an illustration). 
Note that the roots of the subtrees need not be adjacent to $x$ in $G$, hence this search tree forms a rooted tree on the same vertex set as $G$, with a possibly distinct edge set. We will refer to such a tree as a \emph{search tree on the tree} $G$. 
This tree can be searched by iteratively performing node queries, starting from the root: 
at each node $x$, an oracle tells us that either $x$ is the vertex we look for, or in which of the remaining connected components we need to pursue the search.

Note that in the special case where the underlying tree $G$ is a path, a search tree on $G$ is simply a binary search tree on the vertices of $G$ ordered with respect to their location in the path. Search trees on trees are also special cases of trees on graphs, also known as {\em vertex rankings} and they have been studied in various areas of discrete mathematics and computer science, e.g., polyhedral combinatorics, combinatorial optimization, graph theory (we discuss some of the related work in Section~\ref{sec:related_work}).



\paragraph{Adaptive Binary Search Trees --- BST Model of Computation.}
Binary search trees (BSTs) can be viewed as a model of computation where in each operation there is a pointer that starts at the root and can be moved to adjacent nodes at unit cost; additionally a local change known as a \emph{rotation} can be performed at unit cost. A sequence of search operations $X = x_1,\dotsc,x_m$ to nodes of the BST is requested and the goal is to perform all searches at minimum cost. This model was formalized in \cite{Wilber89,DHIP07} and classic binary search trees such as Red-Black trees \cite{DBLP:conf/focs/GuibasS78} and AVL trees \cite{AVL} fit in this model. In the \emph{offline} version of this model, the sequence $X$ is known in advance and the rotations performed might be based on the knowledge of next requests. On the other hand, in the \emph{online} version, each request $x_i$ is revealed after the previous search $x_{i-1}$ has been performed. 

In~\cite{ST85}, Sleator and Tarjan conjectured that there exists a single online BST-model algorithm whose runtime over every sufficiently long sequence of searches is within a constant factor of the optimal \textit{offline} BST-model algorithm on that sequence. This notion of optimality, called \emph{dynamic optimality} is much stronger than other typical types of optimality, including those that are stated in terms of a distribution of operations. 

Sleator and Tarjan~\cite{ST85} introduced the \emph{Splay Tree}, a self-adjusting BST-model algorithm that uses a small set of restructuring heuristics to move the searched node item to the root.
They conjectured that the splay tree was dynamically optimal in the BST model, a conjecture that remains open.
While constant-factor competitiveness remains elusive in the BST model, a breakthrough in the theory of binary search trees was achieved by Demaine, Harmon, Iacono, and P\v{a}tra\c{s}cu~\cite{DHIP07}, who introduced the {\em tango trees}, which were proved to be $O(\log \log n)$-competitive. (See \ref{s:abst} for a discussion of the history of binary search trees in the context of our result).

\paragraph{Adaptive Search Trees on Trees.} Interestingly, while numerous questions on the optimal design of search trees on trees have been addressed (see Section~\ref{sec:related_work}), adaptivity by changing the search tree on tree has never been considered. 
This can be achieved, just in the same way as in binary search trees, via an elementary {\em rotation} operation.
These rotations generalize rotations on binary search trees, as illustrated in Figure~\ref{fig:rotation}.

In this paper, we introduce a natural generalization of the BST computation model involving unit-cost finger moves and rotations, which we call the \emph{general search tree model (GST)}, and consider the design of {\em competitive online search trees on trees} in this model. It is tempting to consider known adaptive binary search tree techniques and try to generalize them to search trees on trees. On the other hand, search trees on trees are much more subtle and many nice properties of BSTs do not generalize. As a result several ``natural'' attempts to generalize the techniques developed for BSTs, seem to fail, as we discuss in detail in Section~\ref{s:compare}. Thus we need several new ideas and techniques in order to design competitive online adaptive search trees.

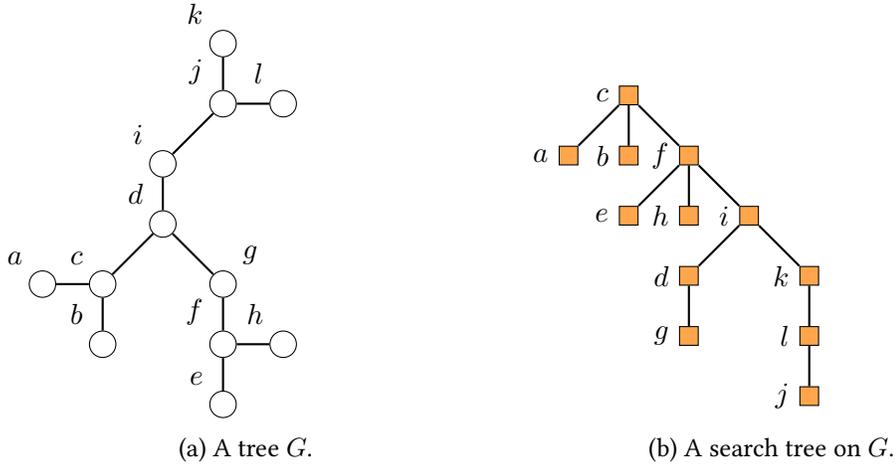
\begin{figure}
  \centering
  \begin{subfigure}[t]{0.4\linewidth}
   \begin{tikzpicture}[scale=.8,auto,swap]
    \foreach \pos/\name in {{(1,3)/a}, {(2,2)/b}, {(2,3)/c}, {(4,2)/f}, {(4,6)/j}, {(3,4)/d}, {(4,1)/e}, {(5,2)/h}, {(3,5)/i}, {(4,7)/k}, {(5,6)/l}}
        \node[vertex, label=above left:$\name$] (\name) at \pos {};
    \node[vertex, label=above right:$g$] (g) at (4,3) {};
    \foreach \source/\dest in {a/c,c/b,c/d,d/g,g/f,f/e,f/h,d/i,i/j,j/k,j/l} \path[edge] (\source) -- (\dest);
   \end{tikzpicture}
   \caption{A tree $G$.}
   \label{fig:treeG}
  \end{subfigure}\quad
  \begin{subfigure}[t]{0.4\linewidth}
   \begin{tikzpicture}[scale=.8,auto,swap]
    \foreach \pos/\name in {{(1,5)/a}, {(2,5)/b}, {(2,6)/c}, {(3,5)/f}, {(3,2)/g}, {(5,2)/l}, {(3,3)/d}, {(2,4)/e}, {(3,4)/h}, {(4,4)/i}, {(5,3)/k}, {(5,1)/j}}
        \node[svertex, label=left:$\name$] (\name) at \pos {};
    \foreach \source/\dest in {a/c,b/c,f/c,e/f,h/f,f/i,i/d,i/k,d/g,k/l,j/l} \path[edge] (\source) -- (\dest);
   \end{tikzpicture}
   \caption{A search tree on $G$.}
   \label{fig:searchtree}       
  \end{subfigure}
  \caption{Illustration of a search tree on a tree $G$.}
  \label{fig:ex_search_tree}
\end{figure}

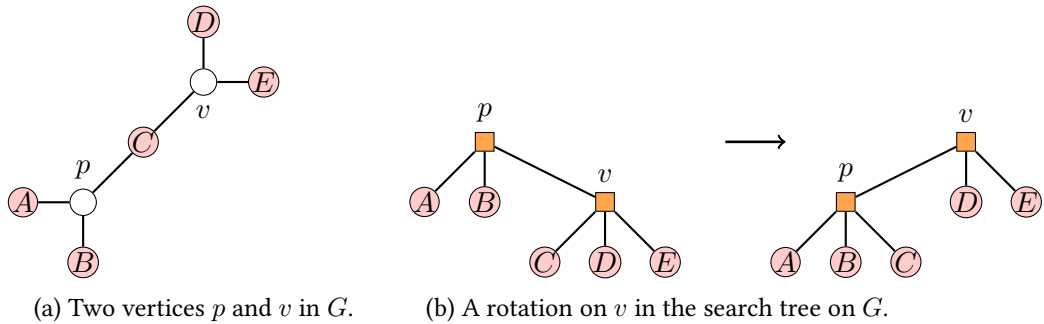
\begin{figure}
  \centering
  \begin{subfigure}[t]{0.3\linewidth}
   \begin{tikzpicture}[scale=.8,auto,swap]
    \foreach \pos/\name in {{(1,2)/A}, {(2,1)/B}, {(3,3)/C}, {(4,5)/D}, {(5,4)/E}} \node[vertex, fill=red!20] (\name) at \pos {$\name$};
    \node[vertex, label=above:$p$] (p) at (2,2) {};
    \node[vertex, label=below:$v$] (v) at (4,4) {};
    \foreach \source/\dest in {A/p,B/p,p/C,C/v,v/D,v/E} \path[edge] (\source) -- (\dest);
   \end{tikzpicture}
   \caption{Two vertices $p$ and $v$ in $G$.}
   \label{fig:rot1}
  \end{subfigure}\quad
  \begin{subfigure}[t]{0.4\linewidth}
   \begin{tikzpicture}[scale=.8,auto,swap]
    \begin{scope}[xshift=0cm]
    \foreach \pos/\name in {{(1,2)/A}, {(2,2)/B}, {(3,1)/C}, {(4,1)/D}, {(5,1)/E}} \node[vertex, fill=red!20] (\name) at \pos {$\name$};
    \node[svertex, label=above:$p$] (p) at (2,3) {};
    \node[svertex, label=above:$v$] (v) at (4,2) {};
    \foreach \source/\dest in {p/A,p/B,p/v,v/C,v/D,v/E} \path[edge] (\source) -- (\dest);
  \end{scope}
    \begin{scope}[xshift=6cm]
    \draw[oriented edge] (0,3) -- (1,3);
    \foreach \pos/\name in {{(1,1)/A}, {(2,1)/B}, {(3,1)/C}, {(4,2)/D}, {(5,2)/E}} \node[vertex, fill=red!20] (\name) at \pos {$\name$};
    \node[svertex, label=above:$p$] (p) at (2,2) {};
    \node[svertex, label=above:$v$] (v) at (4,3) {};
    \foreach \source/\dest in {p/A,p/B,p/v,p/C,v/D,v/E} \path[edge] (\source) -- (\dest);
  \end{scope}
\end{tikzpicture}
   \caption{\label{fig:rotation}A rotation on $v$ in the search tree on $G$.}
  \end{subfigure}
  \caption{Rotations.}
\end{figure}

\subsection{Our Results}

In this work we define a model of computation for search trees on trees, the \emph{general search tree model (GST)}. This model generalizes the BST model, which corresponds to the special case where the underlying tree is a path. We obtain both lower and upper bounds for this model, matching the ones known for the BST model. 

\paragraph{Lower Bound.} We obtain a lower bound on the cost of any algorithm in the GST model, by generalizing the {\em interleave lower bound} of binary search trees~\cite{Wilber89,DHIP07} to search trees on trees.

\paragraph{Upper Bound.} Our main result is an online algorithm for executing search sequences in search trees on trees that is {\em $O(\log \log n)$-competitive}, that is, whose running time in the GST model is at most a factor $O(\log \log n)$ away from that of any other algorithm, {\em even knowing all search requests in advance}. This matches the best competitive factor known so far for binary search trees.

The basic idea is to connect the cost of the algorithm to the interleave lower bound and show that each time this lower bound increases by 1, the algorithm incurs a cost at most $O(\log \log n)$. This is based on the paradigm developed by Demaine, Harmon, Iacono and P\v{a}tra\c{s}cu for Tango trees~\cite{DHIKP09} to achieve competitive ratio $O(\log \log n)$ for binary search trees. However substantially more involved ideas and techniques are needed to obtain the generalization. First, we develop a notion of {\em Steiner-closed search trees} that seems to be independently useful. Moreover the results on BSTs rely crucially on the fact that a subset of $k$ vertices (defining a \textit{preferred path}, see Section~\ref{sec:ub}) can be stored easily in a BST data structure (like red-black trees) that supports split and merge in $O(\log k)$ time. On the other hand, in search trees there is no straightforward analogue. We need a subtle two-level decomposition involving link-cut trees~\cite{ST83,T83} and show that the resulting data structure is a valid search tree on tree. 

\paragraph{Organization.}
In the next subsection, we detail a number of related works. In Section~\ref{s:compare} we provide a comparison between binary search trees and search trees on trees, mentioning the main properties that do not generalize and showing that search trees on trees have a much more rich and subtle structure. In Section~\ref{sec:model}, we give the basic definitions related to search trees and introduce the precise computation model involved.
In Section~\ref{sec:lb}, we show how to generalize the {\em interleave lower bound} of binary search trees to search trees on trees. 
The description of our $O(\log \log n)$-competitive algorithm is given in Section~\ref{sec:ub}.

\subsection{Related Work}
\label{sec:related_work}
\subsubsection{Search trees, vertex rankings, and related notions}
The structure we refer to as search trees on trees, and more generally search trees on graphs, has been studied for decades under many different names. 
They are known as elimination trees~\cite{P88,A94,BGHK95}, vertex rankings~\cite{S89,D93,BDJKKMT98}, and tubings~\cite{CD06,MP15}, among others.

The minimum height of a search tree on a graph $G$ is known in the graph theory community as the {\em tree-depth} of $G$. 
We refer to the text of Ne{\v s}et{\v r}il and Ossona de Mendez~\cite{NO12} for details on the equivalence between the various definitions of tree-depth as well as its connections with other structural parameters of graphs.

Many previous works have been dedicated to constructing minimum-height, and therefore {\em worst-case optimal search trees}~\cite{P88,A94,D93}. 
Most relevant to our result, it was shown by Sch{\"a}ffer that a minimum-height search tree on a tree, hence the tree-depth of a tree, can be found in linear time~\cite{S89}.

\subsubsection{Other search tree models}
More recently, different models were proposed for searching for elements in tree and graph-structured universes. 
Our definition of search trees is based on {\em vertex queries}: every node of the search tree is associated with a vertex of the searched graph.
In contrast, {\em edge queries} have also been studied~\cite{BFN99,OP06,MOW08,CKLPV16,CJLM11,CJLM14}, in which every node of the search tree is associated with an edge of the graph, and the oracle points to one of the two connected components containing the sought vertex. This is sometimes cast as searching in a partial order. A dynamic variant, supporting insertion and deletion of elements, has been proposed by Heeringa et al.~\cite{HIT11}.
Another model that was recently proposed for searching a graph involves an oracle that given a vertex, returns the first edge incident to this vertex on a shortest path to the sought vertex~\cite{EKS16}.

\subsubsection{Polyhedral combinatorics}
The {\em associahedra} are polytopes whose vertices are in bijection with binary search trees on $n$ elements, and edges correspond to pairs of binary search trees that differ by one single rotation~\cite{T51,S63,L89,CSZ15}. From a classical Catalan bijection, their skeletons are also flip graphs of triangulations of a convex polygon on $n+2$ vertices. {\em Graph associahedra} were studied independently by Carr and Devadoss~\cite{CD06}, and Postnikov~\cite{P09}. Just like in associahedra, the vertices of a graph associahedron defined from a graph $G$ are in bijection with the search trees on $G$, and its skeleton is the {\em rotation graph} of the search trees on $G$, that is, two vertices are connected by an edge of the polytope if and only if the search trees differ by a single rotation. The search trees and the rotations are essentially what is referred to by Carr and Devadoss~\cite{CD06} as {\em tubings} and {\em flips}, respectively. In Postnikov~\cite{P09}, and Postnikov, Reiner, Williams~\cite{P08}, the search trees are called {\em $\mathcal B$-trees}. 

Sleator and Tarjan raised the question of the diameter of the associahedron: the largest rotation distance between any two binary search trees on $n$ elements. Together with Thurston, they proved that it was $2n-6$ for sufficiently large values of $n$~\cite{STT86}. A few years ago, Pournin gave a combinatorial proof of this result and established it for every value of $n$~\cite{P14}.
A similar diameter question was studied for graph associahedra by Manneville and Pilaud~\cite{MP15}. They asked the question of the worst-case diameter of {\em tree associahedra}, hence the maximum rotation distance between two search trees on a tree. Cardinal, Langerman, and Perez-Lantero proved that the correct bound was $\Theta (n\log n)$~\cite{CLP18}. 
This body of work provides us with the same theoretical background on search trees on graphs and trees as the combinatorics of associahedra does on binary search trees. 

\subsubsection{Dynamic optimality of binary search trees} \label{s:abst}
Our work most prominently relies on insights and progress on the so-called {\em dynamic optimality} conjecture for binary search trees which posits the existence of $O(1)$-competitive online binary search trees. 
Here we consider the total number of elementary operations, finger moves and rotations, required to search a given sequence of nodes in the tree, and wish this number to lie within a constant factor of that of any other algorithm, even when the latter has access to the sequence of queries in advance. 
We refer to Iacono for a survey on this question~\cite{I13}.

It is famously conjectured that splay trees have the dynamic optimality property. 
This remains unproven, although a number of weaker runtime bounds, such as the static optimality property~\cite{ST85}  (splay trees are as good as the best static BST), the working set property~\cite{ST85} (searching an item is fast if it has been searched recently), and the dynamic finger property~\cite{CMSS00,C00} (searching an item is fast if it is close in key-space to the previous search), are known to hold. 
The crucial {\em Access Lemma} of splay trees and the conditions under which it holds have been extensively analyzed~\cite{S96,GM04,CG0MS15b}.

More recently, the {\em geometric view} on binary search trees was introduced by Demaine, Harmon, Iacono, Kane, and P\v{a}tra\c{s}cu~\cite{DHIKP09}. In this view, optimality of binary search trees is cast as a geometric problem on point sets in the plane. It also introduces the online {\em greedy binary search trees}, an online version of the greedy algorithm introduced independently by Lucas~\cite{L88} and Munro~\cite{Munro00}, conjectured to be dynamically optimal as well. 
Recently, greedy has been proven to satisfy the dynamic finger property~\cite{IL16}, a generalization of static optimality and dynamic finger which is related to the idea of \emph{lazy search} introduced in~\cite{DBLP:journals/algorithmica/BoseDIL16}, as well as certain so-called pattern avoiding search sequences \cite{DBLP:conf/focs/ChalermsookG0MS15}.

Although both splay trees and the greedy algorithm are conjectured to be $O(1)$-competitive, the best known upper bound on their competitive ratio is $O(\log n)$. The best competitive ratio known is $O(\log \log n)$, first achieved by Demaine et. al.~\cite{DHIP07} using {\em tango trees}. Tango trees are designed to approximately match the so-called {\em interleave lower bound}, a variant of a lower bound from Wilber~\cite{Wilber89}. Based on this idea several other $O(\log \log n)$-competitive BST algorithms were later developed~\cite{WDS06,G08}, shown to satisfy several additional properties like the working set bound. It is this structure that we are able to generalize to search trees on trees.

\subsection{From Binary to General Search Trees} \label{s:compare}

It is tempting to think that generalizing the results from the binary search tree model to trees on trees is straightforward. 
It is not. In this section we give several examples of standard and well-known properties of BSTs that do not carry over to search trees on trees and a brief overview of our approach to circumvent the obstacles in the generalization. In fact, given those major differences, it came as a surprise to us that an algorithmic result matching the best known for BSTs can be achieved.

\paragraph{Optimal Static Trees.} Given a probability distribution of the searches constructing a static BST that minimizes the cost of an expected search is a well-studied and frequency taught problem, occupying a section in the ubiquitous CLRS text \cite{DBLP:books/daglib/0023376}. 
A quadratic-time algorithm due to Knuth \cite{DBLP:journals/acta/Knuth71} is a classic application of dynamic programming and a linear-time algorithm proposed in \cite{DBLP:conf/stoc/Fredman75} and fully analyzed in \cite{DBLP:journals/siamcomp/Mehlhorn77} comes within an additive constant of optimal.
However, for trees-on-trees we have no polynomial algorithm or approximation, and attempts to generalize algorithms for trees fail; for example picking as the root the centroid in the distribution and recursing (as in \cite{DBLP:journals/acta/Mehlhorn75}) does not always approximate the optimal tree within a constant factor; dynamic programming also does not have an obvious polynomial-time solution as the number of connected subtrees of a tree is exponential (as opposed to the quadratic number of connected subpaths of a path).

Furthermore, in BST information theory \cite{Shannon1948} applies and it is well-known that the entropy of the distribution is a lower bound on the search cost due. This is not the case for search trees on trees; this is easily seen in, for example a $(n-1)$-star with equal probabilities, by choosing the center of the star as the root every search can be completed in two steps, while entropy gives the higher $\log_2 n$. 
Informally, this is because in binary search trees one comparison gives one bit of information, but in the GST model a comparison gives an amount of information that varies with the degree of the node.

\paragraph{Geometric View.} The {\em geometric view} also constitutes a major contribution to the theory of binary search trees~\cite{DHIKP09}. In this setting, searches to nodes of a binary search trees are pictured as points in the plane, with one dimension being key value and the other time, and any execution of a binary search tree algorithm on this query sequence is pictured as a superset of these points, that satisfy the a geometric condition known as arboral setisfaction. 
The proof of equivalence between this point set and a binary search tree algorithm relies crucially on the fact that any binary search tree is at rotation distance at most linear from any other~\cite{STT86}. 
This is not true for search trees on trees: there exist pairs of search trees on a tree $G$ of $n$ vertices that are at rotation distance $\Omega (n\log n)$~\cite{CLP18}.
It is because of this that extending the geometric view to the trees-on-trees setting fails, even when one dimension is viewed as being tree-shaped.

\paragraph{Splay Trees and Greedy.} One might also ask whether there is a natural generalization of the splay tree algorithm to search trees on trees. Even taking into account the most relaxed conditions for the access lemma to hold~\cite{S96,CG0MS15b}, the question remains unclear due to entropy not being a lower bound.
Additionally, a natural adaption of the reconfiguration heuristic of splay trees to the trees-on-trees setting results in a structure that can be as bad as possible, even on a star graph.
It is also unclear how we can generalize greedy binary search trees --- another candidate for dynamic optimality~\cite{I13} --- to the tree setting as the online variant depends on the transformations inherent in the geometric view.

\paragraph{Our Approach.} It is in this context of the failure of so much of the machinery that has been developed over the past fifty years of BST research to carry over to the GST model that we surprisingly are able to develop a $O(\log \log n)$-competitive algorithm in the GST model. 
Our result is inspired by the BST-model Tango trees, but as the reader shall discover we need a number of new observations, such as the notion of Steiner-closed, which are specific to the GST model.
Crucial to obtaining our result is observing that while entropy-based lower bounds fail in the GST model, we are able to adapt one of the lower bounds due to Wilber \cite{Wilber89} to the GST model. This lower bound is then matched by a factor $O(\log \log n)$ to our data structure using a two-level decomposition. We first decompose a balanced search tree into preferred paths which are represented by search trees. By resorting to Sleator and Tarjan's link-cut trees~\cite{ST83} for handling the changes in preferred paths, we end up with a two-level decomposition into paths, which is eventually managed by splay trees.

%% file: 02-model.tex
\section{ Computation model}
\label{sec:model}

We proceed by defining our rotation-based search tree model.

\begin{defn}[Search tree on a tree]
A rooted tree $T$ is a valid \emph{search tree}
on a given unrooted tree $G = (V,E)$ if
the root $r$ of $T$ stores a vertex of $G$ and
the rooted subtrees of $T\setminus r$ are valid search trees on the connected components of $G \setminus r$.
\end{defn}

Note that in the above definition the trees $T$ and $G$ do not have degree restrictions. While $T$ is rooted, there is no order among the children of a node. Observe that in case the tree $G$ is a path, then a search tree $T$ of $G$ is a binary search tree (BST) with respect to the total order implied by the path. Thus, search trees on trees are a natural generalization of binary search trees.
Throughout the rest of the paper we assume a fixed tree $G$ unless otherwise indicated and $n$ denotes the number of vertices in $G$.

\begin{defn}[Rotation]
A rotation on a non-root node $v$ of $T$ is a local change which yields another search tree constructed as follows:
Let $p$ be the parent of $v$ in $T$. 
Swap $p$ and $v$ in $T$. All children of $p$ remain children of $p$. For a child $u$ of $v$, let $S_u$ be the set of nodes in its subtree. For at most one child $u$ of $v$, there might be a node of $S_u$ adjacent to $p$ in $G$; then $u$ becomes a child of $p$; all other children of $v$ remain children of $v$.
\end{defn}

We refer to Figures~\ref{fig:rotation} and~\ref{fig:rotations_2} for a visual explanation. The following is a direct observation.

\begin{figure}
  \centering
   \begin{tikzpicture}[scale=.8,auto,swap]
    \begin{scope}[xshift=0cm]
    \foreach \pos/\name in {{(1,5)/a}, {(2,5)/b}, {(2,6)/c}, {(3,5)/f}, {(3,2)/g}, {(5,2)/l}, {(3,3)/d}, {(2,4)/e}, {(3,4)/h}, {(4,4)/i}, {(5,3)/k}, {(5,1)/j}}
        \node[svertex, label=left:$\name$] (\name) at \pos {};
    \foreach \source/\dest in {a/c,b/c,f/c,e/f,h/f,i/d,i/k,d/g,k/l,j/l} \path[edge] (\source) -- (\dest);
    \path[thick edge] (f) -- (i);
  \end{scope}
  \begin{scope}[xshift=6cm]
    \draw[oriented edge] (-1,6) -- (0,6);
    \foreach \pos/\name in {{(1,5)/a}, {(2,6)/c}, {(2,5)/b}, {(3,5)/i}, {(2,4)/f}, {(4,4)/k}, {(1,3)/e}, {(2,3)/h}, {(3,3)/d}, {(4,3)/l}, {(3,2)/g}, {(4,2)/j}} \node[svertex, label=left:$\name$] (\name) at \pos {};
    \foreach \source/\dest in {a/c,c/b,c/i,i/k,e/f,f/i,f/h,k/l,d/g,j/l} \path[edge] (\source) -- (\dest);
    \path[thick edge] (f) -- (d);
  \end{scope}
  \begin{scope}[xshift=11cm]
    \draw[oriented edge] (-1,6) -- (0,6);
    \foreach \pos/\name in {{(1,5)/a}, {(2,6)/c}, {(2,5)/b}, {(3,5)/i}, {(2,3)/f}, {(4,4)/k}, {(1,2)/e}, {(2,2)/h}, {(2,4)/d}, {(4,3)/l}, {(3,2)/g}, {(4,2)/j}} \node[svertex, label=left:$\name$] (\name) at \pos {};
    \foreach \source/\dest in {a/c,c/b,c/i,i/k,i/d,e/f,f/h,f/g,k/l,j/l,d/f} \path[edge] (\source) -- (\dest);
  \end{scope}
  \end{tikzpicture}
  \caption{Rotations in search trees on tree $G$ of Figure~\ref{fig:treeG}. The rotation is performed on the blue edge. Thus, for example, in the left tree $i$ is a child of $f$ and in the middle tree $f$ has become a child of $i$.}
  \label{fig:rotations_2}
\end{figure}
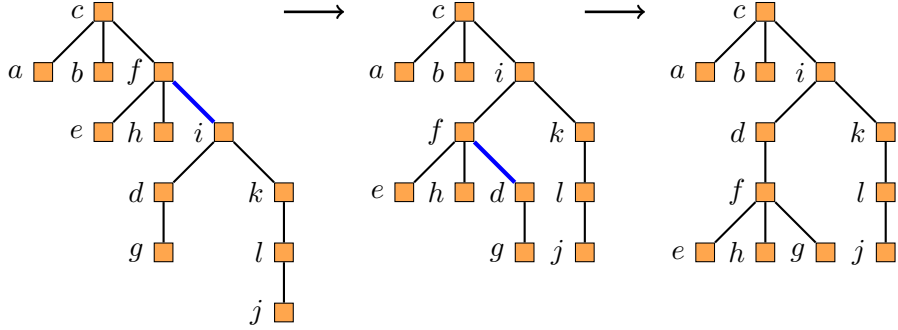

\begin{obs}
If $T$ is a valid search tree on $G$, and $v$ is a nonroot node of $T$, then the tree obtained after a rotation on $v$ is a valid search tree on $G$
\end{obs}

\begin{defn}[GST model of computation]
In the GST model of computation, we are given a tree $G$ and we maintain a tree $T$ which is valid search tree on $G$. At each time, there is a single node pointer at $T$. At unit cost we can perform the following operations:
\begin{enumerate}
\item Move the pointer to a child or the parent of the current node.
\item Rotate the current node $v$.
\end{enumerate} 
A \emph{search} operation for $v\in V$ is any sequence of unit-cost operations where the pointer starts at the root $r$ of $T$ and points to $v$ at some point during the execution of the operation.
\end{defn}

Using this definition, it is easy to see that the GST model is a generalization of the binary search tree (BST) model of computation.

\begin{obs}[Relation with BST model]
\label{obs:gst2bst}
If the tree $G$ is a path, any valid search tree $T$ on $G$ is a binary search tree, search tree rotations correspond to BST rotations, and the GST model of computation is equivalent to the BST model of computation as formulated in \cite{DHIP07}. 
\end{obs}

A sequence $X=x_1, x_2, \ldots x_m$ is a valid search sequence in a tree $G=(V,E)$ if all $x_i \in V$. The main goal of this work is to examine the complexity of executing a search sequence in the GST model.

\begin{defn}[Optimal]
Let $OPT(G,X)$ be the optimal cost of any  GST-model algorithm to execute the sequence of searches $X$ starting from any initial search tree $T$ on $G$.
\end{defn}

Note that while we formally care about rotations and pointer moves only, we are mostly interested in practical algorithms, for which the overall cost, including additional bookkeeping operations, is kept within a reasonable factor of the GST cost.

The notion of preferred child of a node in a search tree will be crucial for both the lower and the upper bound. 

\begin{defn}[Preferred Child]
Let $P$ be a valid search tree on $G$. Let $y$ be a non-leaf node of $P$, with children $y_1,\dotsc,y_d$. At each time $t \in [1,m]$, we define the \textit{preferred child} of $y$ to be the child $y_i$ whose subtree $P(y_i)$ contains the most recent searched vertex in $x_1, \ldots x_t$ that is in a node of $P(y)$ (or is undefined if none of these searches are in $P(y)$). In case last request in $P(y)$ is to $y$, we set preferred child of $y$ to be $y_1$. 
\end{defn}

Note that the preferred child of a node changes throughout the execution of sequence $X$. Understanding and characterizing those changes is an essential part of both the lower bound and the construction of our data structure. 

%% file: 03-lower-bound.tex
\section{Lower bound}
\label{sec:lb}

In this section we give a lower bound on the optimal cost of executing a sequence of searches in vertices of a tree $G$ in the search tree model (GST), by generalizing the interleave lower bound of~\cite{DHIP07} for binary search trees. 

In the following definitions assume a fixed tree $G = (V,E)$ on $n$ vertices and let $X=x_1,x_2, \ldots x_m$ be a request sequence of searches drawn from $V$. We begin by defining the interleave bound.  

\begin{defn}[Interleave Bound]
Let $P$ be a valid search tree on $G$. The interleave bound of a node $y$ of $P$ is the number of times the preferred child of $y$ changes over time $1, 2, \ldots m$.  The interleave bound $I(G,P,X)$ is the sum of the interleave bounds of the nodes. 
\end{defn}

Note that in the definition above, $P$ is a \textit{fixed} search tree and does not change throughout the execution of $X$. We show that the interleave bound value of \emph{any} fixed tree $P$ can be used to obtain a lower bound on $\OPT(G,X)$, the optimal cost to execute $X$ in the GST model. 

\begin{theorem}[Interleave Lower Bound in the GST model]
\label{thm:lower_bound}
Let $P$ be a valid search tree on $G $. For any search sequence $X$ in the GST model of computation, we have that $OPT(G,X) \geq I(G, P,X)/2-n$.
\end{theorem}






\begin{proof}
 Let $\ALG$ be any GST-model algorithm. At a high-level, the proof consists of two main steps:
\begin{enumerate}[(i)]
\item We show that if for a fixed node $y$ in $P$ the interleave bound value is $q$, then there are at least $q/2-1$ unit-cost operations performed by $\ALG$. We charge those operations to node $y$.
 \item We show that for two different nodes $y \neq z$ of $P$, the unit-cost operations charged to $y$ and $z$ are disjoint. 
\end{enumerate}

It is easy to see that those two steps imply the theorem; by summing overall nodes $y$ of $P$, we get that $\ALG$ has cost at least $I(G, P,X)/2 - n $.

\vspace{0.2cm}

\paragraph{Notation and Definitions} Before proving parts (i) and (ii) we introduce some definitions and notation. Let $T_t$ be the tree maintained by $\ALG$ after the $t$th search for $t=1,\dotsc,m$. 

In all following definitions, we fix $y$ to be a node of $P$ of degree $d$ with children $y_1,\dotsc,y_d$ (ordered arbitrarily but consistently throughout the whole execution). Let $P(y)$ denote the subtree of $P$ rooted at $y$. For $i=1,\dotsc,d$, let $\ell_i$ be the node of subtree $P(y_i)$ with the smallest depth in tree $T_t$ (assuming that $y \in P(y_1)$), for any  $1\leq t \leq m$.

\begin{defn}[Dominating node and dominating subtree]
Let $\ell_{i_t}$ be the node with smallest depth in $T_t$ among $\ell_1,\dotsc,\ell_{d}$, for some $1 \leq i_t \leq d$. Then, $\ell_{i_t}$ is the lowest common ancestor in $T_t$ of all nodes stored in $P(y)$. We call $\ell_{i_t}$ the \emph{dominating node of $P(y)$ in $T_t$} and $P(y_{i_t})$ the \emph{dominating subtree} of $P(y)$.
\end{defn}
Note that as the tree $T_t$ evolves over time, the dominating subtree of $y$ might change. We couple this definition with the definition of transition points. 

\begin{defn}[Transition point]
Let $\ell_{i_t}$ be the dominating node of $P(y)$ in $T_t$. For each $i \neq i_t$, we call $\ell_i$ to be the \emph{transition point} of $y$ for $P(y_i)$ at time $t$.
\end{defn}

Note that at each given time $t$ we have exactly $d-1$ transition points of $y$, one corresponding to each $P(y_i)$, for $i \neq i_t$. In the proof of step (i), we will charge $y$ only for touches of $\ALG$ to its transition points.

\paragraph{Properties of transition points.} We now state some basic observations following from the definition of transition points which will be crucial in our proof.


\begin{obs}
\label{obs:transition_opt_1}
A transition point of a node $y \in P$ can not be the root of $T_t$, since $\ell_{i_t}$ is its ancestor. Thus whenever $\ALG$ has to touch a transition point of $y$, it incurs a cost of at least 1.
\end{obs}


\begin{obs}
\label{obs:trans_touch}
Let $\ell_{i_t}$ be the dominating node of $P(y)$ in $T_t$. If the request $x_{t+1}$ is to a node of subtree $P(y_i)$ for some $i \neq i_t$, then the transition point $\ell_{i}$ has to be touched by $\ALG$. 
\end{obs}

\paragraph{Proof of Step (i).} Assume that the interleave bound for $y$ equals $q$. Consider the subsequence of requests $x_{j_1},\dotsc,x_{j_q}$ where the preferred child of $y$ changes. Clearly, any two consecutive requests $x_{j_k},x_{j_{k+1}}$ are from different subtrees $P(y_k),P(y_{k+1})$. 
 
\vspace{0.05cm}
 
 \textit{Requests in non-dominating subtrees:} Each time a node from a non-dominating subtree $P(y_i)$ is requested, the transition node $\ell_i$ has to be touched (by Observation~\ref{obs:trans_touch}). Thus at least one unit-cost operation has to be performed (by Observation~\ref{obs:transition_opt_1}). We charge this operation to $y$. 
 
 \vspace{0.05cm}
 
 \textit{Requests in dominating subtrees:} Let $x_{j_k},x_{j_{k+1}}$ be two consecutive requests such that in both requests a node from the dominating subtree is requested. Since the subtrees $P(y_k)$ and $P(y_{k+1})$ are different, that means the dominating subtree changed at least once during $(j_k,j_{k+1})$; $P(y_{k+1})$ was not a dominating subtree at time $j_k$, but it is at time $j_{k+1}$, thus at some time point during  $(j_k,j_{k+1})$ there should have been a rotation between the lowest common ancestor of points of $P(y_{k+1})$ (which was a transition point of $y$) and the dominating point of $P(y)$. We get that the transition point of $y$ for $P(y_{k+1})$ is touched at least once during $(j_k,j_{k+1})$ and we charge one to this touch. 

\vspace{0.05cm}

\textit{Counting the cost:} Let $q_1$ (resp. $q_2$) be the number of requests to non-dominating (resp. dominating) subtrees of $P(y)$. Note that $q_1 + q_2 = q$. To get an overall bound on the number of unit-cost operations charged to $y$, we perform a simple case analysis on values of $q_1,q_2$. 

In case $q_2 \leq \lceil q/2 \rceil$, we count only the unit-cost operations charged for requests on non-dominating subtrees. There are $q_1$ of them. We have that $q_1 = q - q_2  \geq  q/2-1  $.

On the other hand, in case $q_2 > \lceil q/2 \rceil $, we count all requests charged. There are $q_1$ requests to non-dominating subtrees. Among the requests to dominating subtrees, the ones that are preceded by a request to a non-dominating subtree (there are $q_1$ of them) are not charged. All the others (there are $q_2 - q_1$ of them) are charged. We get that the total charge to $y$ is at least $q_1 + (q_2 - q_1) = q_2 \geq q/2 -1$.

Overall we charged at least $q/2 - 1 $ requests. 

\paragraph{Proof of Step (ii).} We prove the following lemma.
\begin{lemma}
At any given time $t$, each node $v$ of $T_t$ can be a transition node of at most one node $y$ of $P$.
\end{lemma}

\begin{proof}
Take two nodes $y$ and $z$ of $P$. If trees $P(y)$ and $P(z)$ are disjoint, then clearly all transition points of $y$ and $z$ are different. Otherwise, if $P(y)$ and $P(z)$ intersect, one of $y,z$ is an ancestor of the other in $P$. Assume without loss of generality that $y$ is an ancestor of $z$. If the dominating subtree for $y$ is the subtree including $P(z)$, then all transition points of $y$ are not in $P(z)$. Otherwise, there is a transition point $\ell$ for $y$ in the subtree of $P(y)$ which includes $P(z)$. Observe that $\ell$ is the lowest common ancestor of all points of $P(z)$ it $T_t$, so it can not be a transition point for $z$.  \qedhere
\end{proof}

\vspace{0.1cm}

\noindent Since preferred child changes for nodes $y$ of $P$ are charged to touches of transition points for $y$, this implies that by summing overall nodes, no unit cost operation is counted twice. We conclude that the cost of $\ALG$ is at least $I(G,P,X)/2 - n. \qedhere$

\end{proof}


%% file: 04-upper_bound.tex
\section{Tango Trees on Trees}
\label{sec:ub}

In this section, we develop a dynamic search tree data structure that achieves a competitive ratio of $O(\log \log n)$, for search sequences of length $\Omega(n)$. To achieve this, we connect the cost of our algorithm to the \textit{interleave lower bound} for search trees presented in Section~\ref{sec:lb}. 
 
 \paragraph{{\normalfont{\textbf{Preferred paths.}}}} A crucial ingredient in building our data structure is the notion of a \textit{preferred path}. Let $P$ be a fixed valid search tree of a tree $G$. We define a \textit{preferred path} in $P$ as follows: Start from a node that is not the preferred child of its parent (or start from the root) and perform a walk by following the preferred child of the current node, until reaching a leaf. If the preferred child is undefined, pick one arbitrarily.
 
 Note that each change of preferred child during a search sequence results to changes in the preferred paths of $P$. Let $y$ be a node in a preferred path $\Pi$. If $y$ changes preferred child from $y_j$ to $y_{j'}$, then $\Pi$ splits into two paths $\Pi_1$ and $\Pi_2$ where $\Pi_1$ is from the root to $y$ and $\Pi_2$ is rooted at $y_j$ and ends at a leaf. Then, $\Pi_1$ is merged with the preferred path previously rooted at $y_{j'}$

\begin{obs}
\label{obs:preferred_path_changes}
During a search sequence $X$, there are at most $I(G,P,X) +n$ preferred path changes.
\end{obs}
The additive $n$ stems from the fact that when the preferred child of a node $v$ is undefined, we pick one of them arbitrarily in order to form a preferred path. Thus when the preferred child of $v$ is defined for first time, a preferred path change might occur. Over all nodes there are at most $n$ such preferred path changes. 

\paragraph{{\normalfont{\textbf{High-level overview.}}}}
We fix a search tree $P$ of $G$ which we will call \textit{reference tree}, with the property that the height of $P$ is $O(\log n)$.
Note that the reference search tree is used only for the analysis, we do not need to actually store it. 

At a high-level,  we show that for each preferred path of $P$ ``touched'', we can perform search and all update operations (cutting and merging preferred paths) with an overhead factor $O(\log \log n)$. This implies that we have a dynamic search tree execution with cost $O(\log \log n \cdot (I(G,P,X) +n) ) $, due to Observation~\ref{obs:preferred_path_changes}. This combined with Theorem~\ref{thm:lower_bound} implies that the cost of our dynamic search tree data structure is $O(\log \log n) \cdot \OPT(G,X)$ for sufficiently long sequences.

In traditional tango trees, the preferred paths are represented by binary search trees, such as red-black trees~\cite{DHIP07}, or splay trees~\cite{WDS06}. This is possible because the sets of nodes on a preferred path are totally ordered. In contrast, in our case, the sets of nodes on a preferred path no longer are totally ordered. We enforce a property on the nodes of a preferred path that we call the {\em Steiner-closed} property. Surprisingly, this property ensures that the nodes of a preferred path correspond to trees that are obtained as minors of $G$. Thus, to perform the cutting and merging of preferred paths, we need a data structure that allows us to manipulate arbitrary trees. It turns out that Sleator and Tarjan's link-cut trees are precisely what we need~\cite{ST83,T83}.

\paragraph{Roadmap.} The rest of this section is organized as follows.
In subsection~\ref{subsec:steiner}, we introduce the notion of Steiner-closed sets and Steiner-closed trees.
 In subsection~\ref{sec:reftree} we show that there exists a Steiner-closed reference tree $P$ of depth $O(\log n)$, and that the nodes contained in a preferred path form a tree obtained as a minor of $G$.  
 In subsection~\ref{subsec:aux}, we show that the changes of preferred paths can be implemented using link-cut trees. 
 Finally, in subsection~\ref{subsec:alg} we summarize our overall structure and prove that our data structure is $O(\log \log n)$-competitive.

\subsection{Steiner closed sets and trees}
\label{subsec:steiner}

In the following, for two vertices $a$ and $b$ in $G$, let $P(a, b)$ be the the set of vertices on the path from $a$ to $b$.
We first define the notion of convex hull on a set of vertices of a tree. 

\begin{defn}[Convex Hull]
Given a tree $G=(V,E)$, for a set $S \subseteq V$ of vertices, we
define the convex hull $\CH(S)$ be the subgraph of $G$ induced by the
vertices on all paths $P(a, b)$, for all pairs of points $a, b$ in $S$. 
\end{defn}

We now introduce the notion of {\em Steiner-closed} which is critical for our result.

\begin{defn}[Steiner-closed set]
\label{def:steiner-set}
A set $S$ is a {\em Steiner-closed} set of vertices of a tree $G$
provided that every vertex in $\CH(S)\setminus S$ has degree exactly two in $\CH(S)$.
\end{defn}

\begin{defn}[Steiner-closed tree]
\label{def:steiner-closed-tree}
A search tree $T$ of a tree $G$ is a {\em Steiner-closed} tree
provided that the set of nodes on the path in $T$ from the root to an
arbitrary node in $T$ is a Steiner-closed set with respect to $G$.
\end{defn}

The following is one of the key properties of Steiner-closed trees that lets us manipulate them efficiently. Essentially,
we show that there are not many structural changes that occur in $G$ when splitting and merging paths from the root in 
a Steiner-closed search tree $T$ of $G$.

\begin{lemma} \label{lem:jit}
Let $\Pi=p_0, \ldots, p_j$ be a path from the root $p_0$ to a node $p_j$ in a Steiner-closed search tree $T$ of tree $G$. For any $i\in\{1,\ldots, j\}$, let $\Pi'=p_i,\ldots,p_j$. Removing $\CH(\Pi')$ from $\CH(\Pi)$ results in at most 2 connected components.
\end{lemma}
\begin{proof}
Let $\Pi'' = p_0,\ldots, p_{i-1}$. Since $T$ is Steiner-closed, we note that by definition, $\CH(\Pi'')$ is a Steiner-closed set with respect to $G$. For sake of a contradiction, suppose that removing $\CH(\Pi')$ from $\CH(\Pi)$ in $G$ results in at least 3 connected components. Let $C_1, C_2$ and $C_3$ be 3 of these components. Since $\CH(\Pi')$ is a subtree of $\CH(\Pi)$ let $c_ic_i'$ with $i\in\{1,2,3\}$ be the cut edges that connect $C_i$ to $\CH(\Pi')$ with $c_i\in C_i$ and $c_i'\in \CH(\Pi')$. Let $P(c_1,c_2)$ be the path in $\CH(\Pi)$ from $c_1$ to $c_2$ and $P(c_1,c_3)$ the path from $c_1$ to $c_3$. Let $v$ be the first vertex where $P(c_1,c_2)$ and $P(c_1,c_3)$ diverge. Note that $v \not \in \Pi''$. However, $v\in \CH(\Pi'')$ since $c_1, c_2$ and $c_3$ are in $\Pi''$. Moreover, $v$ has degree at least 3 in $\CH(\Pi'')$ which contradicts the fact that $\Pi''$ is a Steiner-closed set in $G$.
\end{proof}


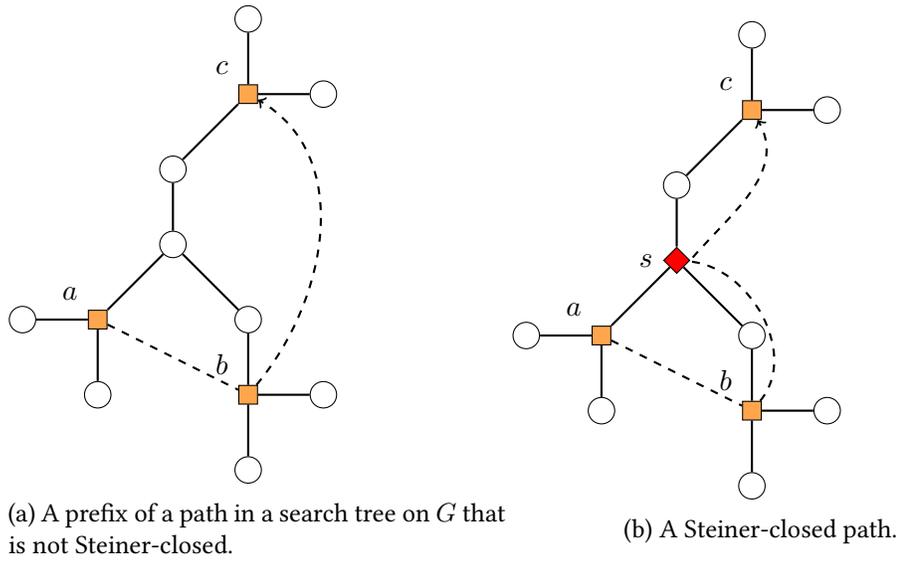
\begin{figure}
\centering
\begin{subfigure}[c]{0.4\textwidth}
  \begin{tikzpicture}[scale=1,auto,swap]
    \foreach \pos/\name in {{(1,3)/a}, {(2,2)/b}, {(4,1)/e}, {(5,2)/h}, {(3,5)/i}, {(4,3)/g}, {(3,4)/d}, {(4,7)/k}, {(5,6)/l}}
        \node[vertex] (\name) at \pos {};
    \foreach \pos/\name/\lab in {{(2,3)/c/a}, {(4,2)/f/b}, {(4,6)/j/c}} \node[svertex,label=above left:$\lab$] (\name) at \pos {};
    \foreach \source/\dest in {a/c,c/b,c/d,d/g,g/f,f/e,f/h,d/i,i/j,j/k,j/l} \path[edge] (\source) -- (\dest);
    \path[dashed oriented edge] (c) to (f) to [out=50,in=-30] (j);
\end{tikzpicture}
\caption{A prefix of a path in a search tree on $G$ that is not Steiner-closed.}
\label{fig:path}
\end{subfigure}
\begin{subfigure}[c]{0.4\textwidth}
  \begin{tikzpicture}[scale=1,auto,swap]
    \foreach \pos/\name in {{(1,3)/a}, {(2,2)/b}, {(4,1)/e}, {(5,2)/h}, {(3,5)/i}, {(4,3)/g}, {(4,7)/k}, {(5,6)/l}} \node[vertex] (\name) at \pos {};
    \node[steiner,label=left:$s$] (d) at (3,4) {};
    \foreach \pos/\name/\lab in {{(2,3)/c/a}, {(4,2)/f/b}, {(4,6)/j/c}} \node[svertex,label=above left:$\lab$] (\name) at \pos {};
    \foreach \source/\dest in {a/c,c/b,c/d,d/g,g/f,f/e,f/h,d/i,i/j,j/k,j/l} \path[edge] (\source) -- (\dest);
    \path[dashed oriented edge] (c) -- (f) to [out=50,in=-5] (d) to [out=50,in=-60] (j);
\end{tikzpicture}
\caption{A Steiner-closed path.}
\label{fig:steinerpath}
\end{subfigure}
\caption{Steiner-closed paths.}
\end{figure}

\begin{figure}
  \centering
  \begin{subfigure}[c]{0.4\linewidth}
   \begin{tikzpicture}[scale=1,auto,swap]
    \foreach \pos/\name in {{(1,3)/a}, {(2,2)/b}, {(2,3)/c}, {(4,2)/f}, {(4,6)/j}, {(3,4)/d}, {(4,1)/e}, {(5,2)/h}, {(3,5)/i}, {(4,7)/k}, {(5,6)/l}}
        \node[vertex, label=above left:$\name$] (\name) at \pos {};
    \node[vertex, label=above right:$g$] (g) at (4,3) {};
    \foreach \source/\dest in {a/c,c/b,c/d,d/g,g/f,f/e,f/h,d/i,i/j,j/k,j/l} \path[edge] (\source) -- (\dest);
   \end{tikzpicture}
   \caption{A tree $G$.}
   \label{fig:treeG3}
  \end{subfigure}\quad
  \begin{subfigure}[c]{0.4\linewidth}
   \begin{tikzpicture}[scale=1,auto,swap]
    \foreach \pos/\name in {{(1,5)/a}, {(2,5)/b}, {(2,6)/c}, {(3,5)/f}, {(3,3)/g}, {(2,4)/e}, {(3,4)/h}, {(5,3)/i}, {(5,0)/l}, {(5,2)/k}}
        \node[svertex, label=left:$\name$] (\name) at \pos {};
    \foreach \pos/\name in {{(4,4)/d}, {(5,1)/j}} \node[steiner, label=left:$\name$] (\name) at \pos {};    
    \foreach \source/\dest in {a/c,b/c,f/c,e/f,h/f,d/f,g/d,d/i,i/k,k/j,l/j} \path[edge] (\source) -- (\dest);
   \end{tikzpicture}
   \caption{A Steiner-closed search tree.}
   \label{fig:steinertree}       
  \end{subfigure}
  \caption{Steiner-closed search tree obtained from the search tree in Figure~\ref{fig:searchtree}.}
\end{figure}
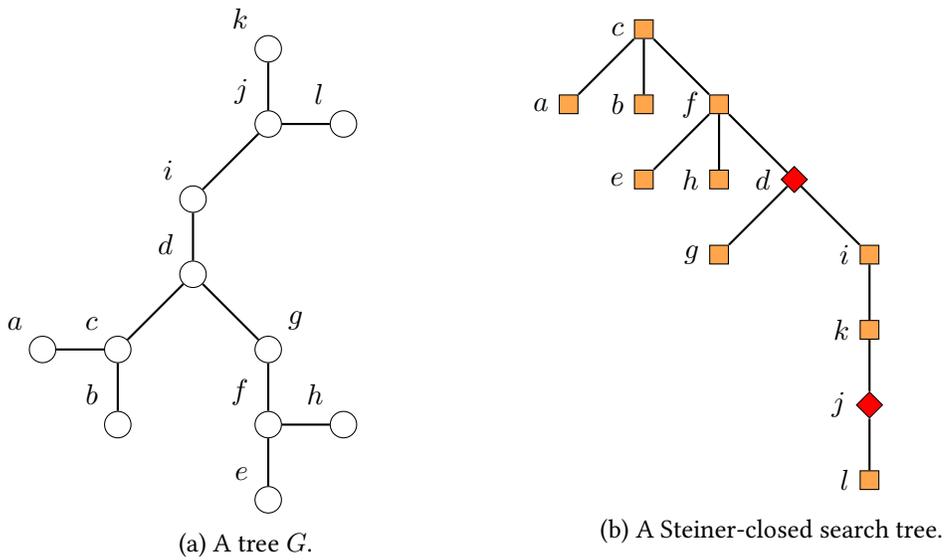

\paragraph{A static \Steiner\ tree.}
Our next lemma shows that given any arbitrary search tree $T$ on a tree $G$, we can transform it to a \Steiner\ tree $T'$ of height at most twice the height of $T$.

\begin{lemma}
\label{lem:steiner_ref}
Given a valid search tree $T$ on a tree $G$, we can create another valid search tree $T'$ of $G$, such that $T'$ is \Steiner\ and $\height(T') \leq 2 \height(T)$ . 
\end{lemma}

\begin{proof}
We show how to transform an arbitrary valid search tree $T$ of a tree $G$ into a \Steiner\ search tree $T'$ where the depth of any node in $T'$ is at most
twice its depth in $T$. We perform a depth-first search on $T$ and build our \Steiner\ tree incrementally. Let $r$ be the root of $T$. For any non-root node $v$ with parent $p(v)$, let $S_v$ be the set of elements in the path from the root $r$ to $v$. At each step $i$ we transform the tree $T_i$ into the tree $T_{i+1}$ such that $T_0 = T$ and $T_{final} = T' $. 

We describe one step of our transformation. Suppose the current tree is $T_i$ and our DFS visits a node $v \in T$ such that: (i) the set $S_{p(v)}$ is Steiner-closed in $G$ and (ii) the set $S_v$ is not Steiner-closed. This means that $\CH(S_v)$ contains a unique vertex $s\in \CH(S_v)\setminus S$ with degree at least 3. Observe that the vertices on the path between $p(v)$ and $v$ in $G$ are contained in the subtree rooted at $v$ in $T_i$. Since $s$ is on this path, it is in the subtree rooted at $v$ in $T_i$. We obtain $T_{i+1}$ by rotating $s$ up the tree until it is between $p(v)$ and $v$, that is, we make it a parent of $v$ and a child of $p(v)$. Note that in $T_{i+1}$ the path from $v$ up to root is now \Steiner\  by construction. We then continue our DFS traversal from node $s$.

Let $v_s$ be the child of $v$ in $T_i$ such that $s$ is in the subtree rooted at $v_s$. Let also $v_1,\dotsc,v_d$ be the other children of $v$. Since we perform single rotations to move $s$ up until it becomes a parent of $v$, it is easy to see that by switching from $T_i$ to $T_{i+1}$, the depth of all nodes in subtrees rooted at $v_1,\dotsc,v_d$ increases by 1 and the depth of all other nodes of $T_i$ does not increase. Thus, for a node $w$ in a subtree rooted at $v_j$, $1 \leq j \leq d$, we have that the depth of $w$ increases by 1, the new root-to-$w$ path is the same as before augmented by node $s$ and the path from root to $v$ is Steiner-closed. This implies that all possible depth increases of $w$ will be caused by nodes in the path between $v$ and $w$. Summing overall changes, we get that for any node of the tree at depth $d$ in $T_{0}$, its depth can increase by 1 at most $d$ times, i.e., its depth at $T'$ is at most $2d$.

At the end of this process, $T'$ is \Steiner\ by construction. Since the depth of every node at most doubles, the height of $T'$ is at most twice the height of $T$. 
\end{proof}

\subsection{Building the Reference tree}
\label{sec:reftree}

To build our reference tree we use as a building block the centroid decomposition $C$ of $G$.

\begin{lemma}
\label{lem:centroid}
Given a tree $G$, there is a search tree $C$ of $G$ with height at most $\log_2 n +1$.
The tree $C$ is a centroid decomposition tree obtained by recursive application of Jordan's theorem~\cite{j1869,harary6graph}: 
Given a tree $G$ with $n$ vertices, there exists a vertex whose removal partitions the tree into components, each with at most $n/2$ vertices.
\end{lemma}

The centroid decomposition of a tree $G$ can be seen as the generalization of balanced binary search trees with height $O(\log n)$. Figure~\ref{fig:centroid_decomposition} shows an example. A centroid decomposition can be computed in time $O(n \log n)$, since we need $O(n)$ to find the centroid of a tree on $n$ vertices.
Using Lemma~\ref{lem:centroid} and Lemma~\ref{lem:steiner_ref} we get the following corollary.

\begin{figure}
  \centering
  \begin{subfigure}[t]{0.4\linewidth}
   \begin{tikzpicture}[scale=.8,auto,swap]
    \foreach \pos/\name in {{(1,3)/a}, {(2,2)/b}, {(2,3)/c}, {(4,2)/f}, {(4,6)/j}, {(3,4)/d}, {(4,1)/e}, {(5,2)/h}, {(3,5)/i}, {(4,7)/k}, {(5,6)/l}}
        \node[vertex, label=above left:$\name$] (\name) at \pos {};
    \node[vertex, label=above right:$g$] (g) at (4,3) {};
    \foreach \source/\dest in {a/c,c/b,c/d,d/g,g/f,f/e,f/h,d/i,i/j,j/k,j/l} \path[edge] (\source) -- (\dest);
   \end{tikzpicture}
   \caption{A tree $G$.}
   \label{fig:treeG2}
  \end{subfigure}\quad
  \begin{subfigure}[t]{0.5\linewidth}
   \begin{tikzpicture}[scale=.8,auto,swap]
    \foreach \pos/\name in {{(1,1)/a}, {(3,1)/b}, {(2,2)/c}, {(4,1)/e}, {(5,2)/f}, {(5,1)/g}, {(6,1)/h}, {(7,1)/i}, {(8,1)/k}, {(9,1)/l}}
       \node[svertex, label=left:$\name$] (\name) at \pos {};
    \node[svertex, label=above left:$d$] (d) at (5,3) {};
    \node[svertex, label=above left:$j$] (j) at (8,2) {};
    \foreach \source/\dest in {a/c,b/c,c/d,d/f,e/f,g/f,h/f,i/j,d/j,j/k,j/l} \path[edge] (\source) -- (\dest);
   \end{tikzpicture}
   \caption{The centroid decomposition of $G$.}
   \label{fig:centroidtree}       
  \end{subfigure}
  \caption{Centroid decomposition.}
  \label{fig:centroid_decomposition}
\end{figure}
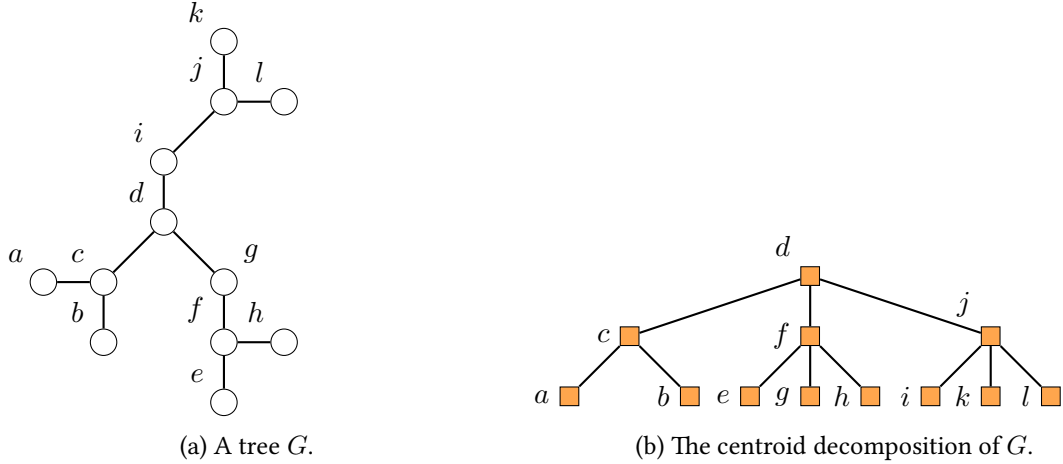

\begin{corol}
\label{ref:cor-reference}
For any tree $G$ on $n$ vertices, there exists a valid search tree $P$ on $G$ which is \Steiner\ and it has height at most $2 \log n +2$.
\end{corol}

This is because the centroid decomposition $C$ has height $\log n +1$, thus the tree $P$ can be obtained by applying Lemma~\ref{lem:steiner_ref} for $T=C$. The tree $P$ will be our reference tree in the rest of this paper. We need to show that each preferred path of $P$ can be stored in a search tree. To do that, we first observe that all preferred paths of $P$ are \Steiner. 

\begin{obs}
\label{obs:Steiner_subtree}
If a search tree $T$ of a tree $G$ is \Steiner, then for all nodes $v$ in $T$, the subtree $T_v$ rooted at $v$ is also \Steiner. 
\end{obs}

This comes from the fact that the subtree rooted at $v$ contains a connected component obtained after the removal of vertices in the path from the root of $T$ to the parent of $v$. Since $T$ is \Steiner, then  $T_v$ should also be \Steiner. This observation implies that all preferred paths of $P$ are \Steiner. The next lemma shows that a \Steiner\ set $S$ corresponds to a tree on $S$ which is a obtained by contracting edges of $G$. 
Let $\bar{P}(a,b)$ denote the set of nodes $v \neq \lbrace a,b \rbrace$ of the path from $a$ to $b$.

\begin{defn}
\label{def:GS}
For a Steiner-closed set of vertices $S$ of $G$, let $G(S)$ to be the graph with vertex set $S$ where two vertices
$a,b\in S$ are connected by an edge if and only if no $c\in S$ is in $\bar{P}(a,b)$.
\end{defn}

Our next Lemma directly follows from Definition \ref{def:steiner-set}. 
\begin{lemma}
\label{lem:search_tree_preferred}
For any \Steiner\ set $S$, $G(S)$ is a tree.
\end{lemma}

It now remains to show how to maintain the trees $G(S)$ corresponding to the preferred paths of $P$.

\subsection{Maintaining preferred paths with link-cut trees}
\label{subsec:aux}

 We now show that a collection of Steiner-closed preferred paths, each of which has length at most $k$, can be stored in a data structure supporting search, cut and merge at a cost of $O(\log k)$ in the GST model of computation. Our goal is to use this for the preferred paths of $P$ of length $k = O(\log n)$, to get the $O(\log k) = O(\log \log n)$ bound.

During an execution of a search sequence we need to perform the following operations on preferred paths:
\begin{enumerate}[(i)]
 \item Search for a node in a preferred path $\Pi$.
 \item Cut a preferred path $\Pi$ into two paths, one consisting of nodes of depth smaller than $d$ in $P$ and the other of nodes of depth at least $d$. We denote this operation $\cut(\Pi,d)$.
 \item Merge two preferred paths $\Pi_1$ and $\Pi_2$, where the bottom node of $\Pi_1$ is the parent of the top node of $\Pi_2$.
\end{enumerate}

Let $\Pi$ be a preferred path containing a \Steiner\ set of nodes $S$. In order to split $\Pi$ into two paths, we split $G(S)$ into two trees $G(S_1)$ and $G(S_2)$, where $S_1$ and $S_2$ are the nodes in $\Pi_1$ and $\Pi_2$. From Observation~\ref{obs:Steiner_subtree}, we know that the bottom part $\Pi_2$ is also \Steiner, which implies $G(S_2)$ is a tree. The path $\Pi_1$ is clearly \Steiner\, since $P$ is \Steiner.
In fact, from the \Steiner\ property and Lemma~\ref{lem:jit}, we have that $G(S_2)$ can be obtained from $G(S)$ by cutting at most two edges of $G(S)$.

Conversely, if we merge two preferred paths $\Pi_1$ and $\Pi_2$ into one path $\Pi$, where the bottom of $\Pi_1$ is connected to the top of $\Pi_2$ in $P$, we need to construct the tree $G(S)$, where $S$ is the union of the sets of nodes $S_1$ and $S_2$ in the paths $\Pi_1$ and $\Pi_2$. Again $G(S_2)$ can be merged with $G(S_1)$ to yield $G(S)$ by cutting $G(S_1)$ at at most two places, and linking the two trees by adding two edges.

\paragraph{Basic operations that need to be supported in logarithmic time.} 
We need to implement a data structure supporting the above operations on the forest of trees $G(S)$ at $O(\log k)$ cost in the GST model.
Each of these operations can be split into a constant number of one of these two operations: {\em cut} a tree into two by removing an edge, and {\em link} two trees into one by adding an edge. In what follows, we refer to the trees $G(S)$, for each set $S$ of nodes in a preferred path of $P$, as the {\em represented trees}.

Here we resort to the classical link-cut trees data structure from Sleator and Tarjan\cite{ST83}. Link-cut trees can be used to implement link and cut operations on a forest of unrooted trees in amortized time proportional to the logarithm of the size of the largest tree in the forest. We refer to Tarjan's textbook for details~\cite{T83}.



The link-cut tree data structure implements a heavy-path decomposition on the represented trees. Each heavy path in this decomposition is in turn represented by a splay tree~\cite{ST85}. Note that this decomposition should not be confused with the decomposition of $P$ into preferred paths; we are performing two levels of decomposition into paths. Link-cut trees effectively reduce the problem of maintaining the forest of trees to binary search tree operations. Our data structure eventually consists of a hierarchy of splay trees, each representing a path in a tree $G(S)$, which in turn corresponds to a path in the reference tree $P$.

We now have to check that the whole structure is indeed a search tree on $G$, and that the binary search tree operations are indeed proper elementary operations in the GST model.  
This we can observe by considering the first preferred path $\Pi$ in $P$ with nodes $S$, and the first heavy path in the decomposition of $G(S)$, stored as a splay tree. Searching in such a splay tree amounts to searching along a path of $G(S)$, whose convex hull is a path in $G$. Recalling our Observation~\ref{obs:gst2bst}, this is a proper search in the GST model. The search now proceeds by searching the next heavy paths in $G(S)$, and maybe switch to other preferred paths of $P$. Overall, this is a proper search in $G$. Similarly, rotations in the splay trees are rotations of the search tree on $G$ as defined in the GST model.

\subsection{Our Algorithm}
\label{subsec:alg}

Let us first sum up the overall data structure.
Given the graph $G$, we construct a balanced \Steiner\ search tree $P$ on $G$, which we refer to as the reference tree. We dynamically maintain a decomposition of $P$ into preferred paths. Each such preferred path with nodes $S$ corresponds to an unrooted tree $G(S)$, which is a minor of $G$. As searches are performed, preferred paths are updated, and these updates correspond to linking and cutting trees $G(S)$. For this, we use link-cut trees. Those in turn decompose the trees $G(S)$ into paths and reduce the operations to link and cut on paths. These operations can be handled by splay trees. Together, they form a search tree on $G$.

\paragraph{{\normalfont{\textbf{Bounding the cost.}}}} We now compare the cost of our \strucname\ data structure to $\OPT(G,X)$. The following lemma makes the essential connection between the number of preferred paths touched during a search and the cost of our algorithm.

\begin{lemma}
\label{lem:tango_wilber1}
Let $\ell$ be the number of preferred child changes during a search. Then the cost of this search is $O((\ell+1)(1+\log \log n))$.  
\end{lemma}
\begin{proof}
During the search, the pointer touches exactly $\ell+1$ preferred paths. We account separately for the search cost and the update cost. 

For each preferred path touched, the search cost is $O(\lceil \log \log n \rceil)$. Thus the total search cost is clearly $O((\ell+1)(1+\log \log n))$. 

We now account for the update cost. Recall that we can cut and merge preferred paths on $k$ nodes in time $O(1 +\log k)$. Since each preferred path has at most $O(\log n)$ nodes, we can perform those updates in time $O(1 +\log \log n)$. There are $\ell$ preferred path changes, and for each change we perform one cut and and one merge operation, we get that the total time for merging and cutting is $O(\ell \cdot (1 +\log \log n))$.  The lemma follows. \end{proof}

We now combine this lemma with Theorem~\ref{thm:lower_bound} to get the competitive ratio of \strucname. 

\begin{theorem}
\label{thm:tango_comp}
For any search sequence of length $m=\Omega(n )$, \strucname\ are $O(\log \log n)$-competitive.
\end{theorem}
\begin{proof}
We account only for the cost occurred during searches, since the cost of transforming the input tree into a valid \strucname\ is just a fixed additive term which does not depend on the input sequence.

By Observation~\ref{obs:preferred_path_changes} we have that the total number of preferred path changes is at most $I(G,P,X)+n$. Using Lemma~\ref{lem:tango_wilber1} and summing up over all search requests, we get that the cost of \strucname\ is $O((I(G,P,X)+n+m)(1+\log \log n))$. 

By Theorem~\ref{thm:lower_bound} this is bounded by $O((\OPT(G,X)+n+m) \cdot (1+\log \log n) ) $. Note that $\OPT(G,X) \geq m$. Since $m=\Omega(n)$, we have that $  \OPT(G,X) + n +m = O( \OPT(G,X))$. We get that the total cost is upper bounded by \[  O \lrp{ \OPT(G,X) \cdot (1 + \log \log n) }. \qedhere \] 
\end{proof}

%% file: 00-main.bbl
\newcommand{\etalchar}[1]{$^{#1}$}
\begin{thebibliography}{CGK{\etalchar{+}}15b}

\bibitem[AH94]{A94}
Bengt Aspvall and Pinar Heggernes.
\newblock Finding minimum height elimination trees for interval graphs in
  polynomial time.
\newblock {\em BIT Numerical Mathematics}, 34(4):484--509, Dec 1994.

\bibitem[AVL62]{AVL}
G.~M. Adel’son-Vel’skii and E.~M. Landis.
\newblock An algorithm for the organization of information.
\newblock {\em Soviet Mathematics Doklady}, 3(5):1259–1263, 1962.

\bibitem[BDIL16]{DBLP:journals/algorithmica/BoseDIL16}
Prosenjit Bose, Karim Dou{\"{\i}}eb, John Iacono, and Stefan Langerman.
\newblock The power and limitations of static binary search trees with lazy
  finger.
\newblock {\em Algorithmica}, 76(4):1264--1275, 2016.

\bibitem[BDJ{\etalchar{+}}98]{BDJKKMT98}
Hans~L. Bodlaender, Jitender~S. Deogun, Klaus Jansen, Ton Kloks, Dieter
  Kratsch, Haiko M{\"{u}}ller, and Zsolt Tuza.
\newblock Rankings of graphs.
\newblock {\em {SIAM} J. Discrete Math.}, 11(1):168--181, 1998.

\bibitem[BFN99]{BFN99}
Yosi Ben{-}Asher, Eitan Farchi, and Ilan Newman.
\newblock Optimal search in trees.
\newblock {\em {SIAM} J. Comput.}, 28(6):2090--2102, 1999.

\bibitem[BGHK95]{BGHK95}
Hans~L. Bodlaender, John~R. Gilbert, Hj{\'{a}}lmtyr Hafsteinsson, and Ton
  Kloks.
\newblock Approximating treewidth, pathwidth, frontsize, and shortest
  elimination tree.
\newblock {\em J. Algorithms}, 18(2):238--255, 1995.

\bibitem[CD06]{CD06}
Michael Carr and Satyan~L. Devadoss.
\newblock Coxeter complexes and graph-associahedra.
\newblock {\em Topology and its Applications}, 153(12):2155--2168, 2006.

\bibitem[CGK{\etalchar{+}}15a]{DBLP:conf/focs/ChalermsookG0MS15}
Parinya Chalermsook, Mayank Goswami, L{\'{a}}szl{\'{o}} Kozma, Kurt Mehlhorn,
  and Thatchaphol Saranurak.
\newblock Pattern-avoiding access in binary search trees.
\newblock In {\em {IEEE} 56th Annual Symposium on Foundations of Computer
  Science, {FOCS} 2015, Berkeley, CA, USA, 17-20 October, 2015}, pages
  410--423, 2015.

\bibitem[CGK{\etalchar{+}}15b]{CG0MS15b}
Parinya Chalermsook, Mayank Goswami, L{\'{a}}szl{\'{o}} Kozma, Kurt Mehlhorn,
  and Thatchaphol Saranurak.
\newblock Self-adjusting binary search trees: What makes them tick?
\newblock In {\em Algorithms - {ESA} 2015 - 23rd Annual European Symposium,
  Patras, Greece, September 14-16, 2015, Proceedings}, pages 300--312, 2015.

\bibitem[CJLM11]{CJLM11}
Ferdinando Cicalese, Tobias Jacobs, Eduardo~Sany Laber, and Marco Molinaro.
\newblock On the complexity of searching in trees and partially ordered
  structures.
\newblock {\em Theor. Comput. Sci.}, 412(50):6879--6896, 2011.

\bibitem[CJLM14]{CJLM14}
Ferdinando Cicalese, Tobias Jacobs, Eduardo~Sany Laber, and Marco Molinaro.
\newblock Improved approximation algorithms for the average-case tree searching
  problem.
\newblock {\em Algorithmica}, 68(4):1045--1074, 2014.

\bibitem[CKL{\etalchar{+}}16]{CKLPV16}
Ferdinando Cicalese, Bal{\'{a}}zs Keszegh, Bernard Lidick{\'{y}},
  D{\"{o}}m{\"{o}}t{\"{o}}r P{\'{a}}lv{\"{o}}lgyi, and Tom{\'{a}}s Valla.
\newblock On the tree search problem with non-uniform costs.
\newblock {\em Theor. Comput. Sci.}, 647:22--32, 2016.

\bibitem[CLP18]{CLP18}
Jean Cardinal, Stefan Langerman, and Pablo P{\'{e}}rez{-}Lantero.
\newblock On the diameter of tree associahedra.
\newblock {\em Electr. J. Comb.}, 25(4):P4.18, 2018.

\bibitem[CLRS09]{DBLP:books/daglib/0023376}
Thomas~H. Cormen, Charles~E. Leiserson, Ronald~L. Rivest, and Clifford Stein.
\newblock {\em Introduction to Algorithms, 3rd Edition}.
\newblock {MIT} Press, 2009.

\bibitem[CMSS00]{CMSS00}
Richard Cole, Bud Mishra, Jeanette~P. Schmidt, and Alan Siegel.
\newblock On the dynamic finger conjecture for splay trees. part {I:} splay
  sorting log n-block sequences.
\newblock {\em {SIAM} J. Comput.}, 30(1):1--43, 2000.

\bibitem[Col00]{C00}
Richard Cole.
\newblock On the dynamic finger conjecture for splay trees. part {II:} the
  proof.
\newblock {\em {SIAM} J. Comput.}, 30(1):44--85, 2000.

\bibitem[CSZ15]{CSZ15}
Cesar Ceballos, Francisco Santos, and G{\"u}nter~M. Ziegler.
\newblock Many non-equivalent realizations of the associahedron.
\newblock {\em Combinatorica}, 35(5):513--551, Oct 2015.

\bibitem[DHI{\etalchar{+}}09]{DHIKP09}
Erik~D. Demaine, Dion Harmon, John Iacono, Daniel~M. Kane, and Mihai
  P\v{a}tra\c{s}cu.
\newblock The geometry of binary search trees.
\newblock In {\em Proceedings of the Twentieth Annual {ACM-SIAM} Symposium on
  Discrete Algorithms, {SODA} 2009, New York, NY, USA, January 4-6, 2009},
  pages 496--505, 2009.

\bibitem[DHIP07]{DHIP07}
Erik~D. Demaine, Dion Harmon, John Iacono, and Mihai P\v{a}tra\c{s}cu.
\newblock Dynamic optimality --- almost.
\newblock {\em {SIAM} J. Comput.}, 37(1):240--251, 2007.

\bibitem[DKKM99]{D93}
Jitender~S. Deogun, Ton Kloks, Dieter Kratsch, and Haiko M{\"u}ller.
\newblock On the vertex ranking problem for trapezoid, circular-arc and other
  graphs.
\newblock {\em Discrete Applied Mathematics}, 98(1):39--63, 1999.

\bibitem[EKS16]{EKS16}
Ehsan Emamjomeh{-}Zadeh, David Kempe, and Vikrant Singhal.
\newblock Deterministic and probabilistic binary search in graphs.
\newblock In {\em Proceedings of the 48th Annual {ACM} {SIGACT} Symposium on
  Theory of Computing, {STOC} 2016, Cambridge, MA, USA, June 18-21, 2016},
  pages 519--532, 2016.

\bibitem[Fre75]{DBLP:conf/stoc/Fredman75}
Michael~L. Fredman.
\newblock Two applications of a probabilistic search technique: Sorting x + y
  and building balanced search trees.
\newblock In {\em Proceedings of the 7th Annual {ACM} Symposium on Theory of
  Computing, May 5-7, 1975, Albuquerque, New Mexico, {USA}}, pages 240--244,
  1975.

\bibitem[Geo08]{G08}
George~F. Georgakopoulos.
\newblock Chain-splay trees, or, how to achieve and prove
  loglogn-competitiveness by splaying.
\newblock {\em Inf. Process. Lett.}, 106(1):37--43, 2008.

\bibitem[GM04]{GM04}
George~F. Georgakopoulos and David~J. McClurkin.
\newblock Generalized template splay: {A} basic theory and calculus.
\newblock {\em Comput. J.}, 47(1):10--19, 2004.

\bibitem[GS78]{DBLP:conf/focs/GuibasS78}
Leonidas~J. Guibas and Robert Sedgewick.
\newblock A dichromatic framework for balanced trees.
\newblock In {\em 19th Annual Symposium on Foundations of Computer Science, Ann
  Arbor, Michigan, USA, 16-18 October 1978}, pages 8--21, 1978.

\bibitem[Har69]{harary6graph}
Frank Harary.
\newblock Graph theory, 1969.

\bibitem[HIT11]{HIT11}
Brent Heeringa, Marius~Catalin Iordan, and Louis Theran.
\newblock Searching in dynamic tree-like partial orders.
\newblock In {\em Algorithms and Data Structures - 12th International
  Symposium, {WADS} 2011, New York, NY, USA, August 15-17, 2011. Proceedings},
  pages 512--523, 2011.

\bibitem[Iac13]{I13}
John Iacono.
\newblock In pursuit of the dynamic optimality conjecture.
\newblock In {\em Space-Efficient Data Structures, Streams, and Algorithms -
  Papers in Honor of J. Ian Munro on the Occasion of His 66th Birthday}, pages
  236--250, 2013.

\bibitem[IL16]{IL16}
John Iacono and Stefan Langerman.
\newblock Weighted dynamic finger in binary search trees.
\newblock In {\em Proceedings of the Twenty-Seventh Annual {ACM-SIAM} Symposium
  on Discrete Algorithms, {SODA} 2016, Arlington, VA, USA, January 10-12,
  2016}, pages 672--691, 2016.

\bibitem[Jor69]{j1869}
Camille Jordan.
\newblock Sur les assemblages de lignes.
\newblock {\em Journal f{\"u}r die reine und angewandte Mathematik},
  70:185--190, 1869.

\bibitem[Knu71]{DBLP:journals/acta/Knuth71}
Donald~E. Knuth.
\newblock Optimum binary search trees.
\newblock {\em Acta Inf.}, 1:14--25, 1971.

\bibitem[Lee89]{L89}
Carl~W. Lee.
\newblock The associahedron and triangulations of the n-gon.
\newblock {\em European Journal of Combinatorics}, 10(6):551--560, 1989.

\bibitem[Luc88]{L88}
Joan M~. Lucas.
\newblock Canonical forms for competitive binary search tree algorithms.
\newblock DGS-TR-250, Department of Computer Science, Hill Center for the
  Mathematical Sciences Busch Campus, Rutgers University, 1988.

\bibitem[Meh75]{DBLP:journals/acta/Mehlhorn75}
Kurt Mehlhorn.
\newblock Nearly optimal binary search trees.
\newblock {\em Acta Inf.}, 5:287--295, 1975.

\bibitem[Meh77]{DBLP:journals/siamcomp/Mehlhorn77}
Kurt Mehlhorn.
\newblock A best possible bound for the weighted path length of binary search
  trees.
\newblock {\em {SIAM} J. Comput.}, 6(2):235--239, 1977.

\bibitem[MOW08]{MOW08}
Shay Mozes, Krzysztof Onak, and Oren Weimann.
\newblock Finding an optimal tree searching strategy in linear time.
\newblock In {\em Proceedings of the Nineteenth Annual {ACM-SIAM} Symposium on
  Discrete Algorithms, {SODA} 2008, San Francisco, California, USA, January
  20-22, 2008}, pages 1096--1105, 2008.

\bibitem[MP15]{MP15}
Thibault {Manneville} and Vincent {Pilaud}.
\newblock {Graph properties of graph associahedra}.
\newblock {\em S{\'e}minaire Lotharingien de Combinatoire}, {B73d}, 2015.

\bibitem[Mun00]{Munro00}
J.~Ian Munro.
\newblock On the competitiveness of linear search.
\newblock In {\em {ESA}}, volume 1879 of {\em Lecture Notes in Computer
  Science}, pages 338--345. Springer, 2000.

\bibitem[NOdM12]{NO12}
Jaroslav Ne{\v s}et{\v r}il and Patrice Ossona~de Mendez.
\newblock {\em Sparsity: Graphs, Structures, and Algorithms}, chapter~6, pages
  115--144.
\newblock Springer, 2012.

\bibitem[OP06]{OP06}
Krzysztof Onak and Pawel Parys.
\newblock Generalization of binary search: Searching in trees and forest-like
  partial orders.
\newblock In {\em 47th Annual {IEEE} Symposium on Foundations of Computer
  Science {(FOCS} 2006), 21-24 October 2006, Berkeley, California, USA,
  Proceedings}, pages 379--388, 2006.

\bibitem[Pos09]{P09}
Alexander Postnikov.
\newblock Permutohedra, associahedra, and beyond.
\newblock {\em International Mathematics Research Notices}, 2009(6):1026--1106,
  2009.

\bibitem[Pot88]{P88}
Alex Pothen.
\newblock The complexity of optimal elimination trees.
\newblock Tech. Report {CS-88-13}, Pennsylvania State University, 1988.

\bibitem[Pou14]{P14}
Lionel Pournin.
\newblock The diameter of associahedra.
\newblock {\em Advances in Mathematics}, 259:13--42, 2014.

\bibitem[PRW08]{P08}
Alex Postnikov, Victor Reiner, and Lauren Williams.
\newblock Faces of generalized permutohedra.
\newblock {\em Documenta Mathematica}, 13:207--273, 2008.

\bibitem[Sch89]{S89}
Alejandro~A. Sch{\"a}ffer.
\newblock Optimal node ranking of trees in linear time.
\newblock {\em Information Processing Letters}, 33(2):91--96, 1989.

\bibitem[Sha48]{Shannon1948}
Claude~Elwood Shannon.
\newblock A mathematical theory of communication.
\newblock {\em The Bell System Technical Journal}, 27(3):379--423, 7 1948.

\bibitem[ST83]{ST83}
Daniel~Dominic Sleator and Robert~Endre Tarjan.
\newblock A data structure for dynamic trees.
\newblock {\em J. Comput. Syst. Sci.}, 26(3):362--391, 1983.

\bibitem[ST85]{ST85}
Daniel~Dominic Sleator and Robert~Endre Tarjan.
\newblock Self-adjusting binary search trees.
\newblock {\em J. {ACM}}, 32(3):652--686, 1985.

\bibitem[Sta63]{S63}
James~Dillon Stasheff.
\newblock Homotopy associativity of {H}-spaces. {I}.
\newblock {\em Transactions of the American Mathematical Society},
  108(2):275--292, 1963.

\bibitem[STT86]{STT86}
Daniel~Dominic Sleator, Robert~Endre Tarjan, and William~P. Thurston.
\newblock Rotation distance, triangulations, and hyperbolic geometry.
\newblock In {\em Proceedings of the 18th Annual {ACM} Symposium on Theory of
  Computing, May 28-30, 1986, Berkeley, California, {USA}}, pages 122--135,
  1986.

\bibitem[Sub96]{S96}
Ashok Subramanian.
\newblock An explanation of splaying.
\newblock {\em J. Algorithms}, 20(3):512--525, 1996.

\bibitem[Tam51]{T51}
Dov Tamari.
\newblock Mono{\"i}des pr{\'e}ordonn{\'e}s et cha{\^i}nes de {M}alcev.
\newblock Th{\`e}se de Math{\'e}matiques, Paris, 1951.

\bibitem[Tar83]{T83}
Robert~Endre Tarjan.
\newblock {\em Data Structures and Network Algorithms}.
\newblock Society for Industrial and Applied Mathematics, Philadelphia, PA,
  USA, 1983.

\bibitem[WDS06]{WDS06}
Chengwen~Chris Wang, Jonathan Derryberry, and Daniel~Dominic Sleator.
\newblock \emph{O}(log log \emph{n})-competitive dynamic binary search trees.
\newblock In {\em Proceedings of the Seventeenth Annual {ACM-SIAM} Symposium on
  Discrete Algorithms, {SODA} 2006, Miami, Florida, USA, January 22-26, 2006},
  pages 374--383, 2006.

\bibitem[Wil89]{Wilber89}
Robert~E. Wilber.
\newblock Lower bounds for accessing binary search trees with rotations.
\newblock {\em {SIAM} J. Comput.}, 18(1):56--67, 1989.

\end{thebibliography}
